\documentclass[aps,chaos,12pt,preprintnumbers,amsmath,amssymb,groupedaddress,superscriptaddress]{revtex4-1}
  
\usepackage{mathptmx} 

\usepackage{graphicx}
\usepackage{amssymb,amsmath}
\usepackage{amstext,amsthm,amssymb,amsfonts}
\usepackage[usenames,dvipsnames]{xcolor}

\newtheorem{theorem}{Theorem}[section]
\newtheorem{lemma}[theorem]{Lemma}

\newtheorem{definition}[theorem]{Definition}

\begin{document}
\date{\today}
\title{When a `rat race' implies an intergenerational wealth trap}

\author{Joel Nishimura}
\affiliation{School of Mathematical and Natural Sciences, Arizona State University, Glendale, AZ 85306, USA}

\begin{abstract}
Two critical questions about intergenerational outcomes are: one, whether significant barriers or traps exist between different social or economic strata; and two, the extent to which intergenerational outcomes do (or can be used to) affect individual investment and consumption decisions.  We develop a model to explicitly relate these two questions, and prove the first such `rat race' theorem, showing that a fundamental relationship exists between high levels of individual investment and the existence of a wealth trap, which traps otherwise identical agents at a lower level of wealth. Our simple model of intergenerational wealth dynamics involves agents which balance current consumption with investment in a single descendant. Investments then determine descendant wealth via a potentially nonlinear and discontinuous competitiveness function about which we do not make concavity assumptions. From this model we demonstrate how to infer such a competitiveness function from investments, along with geometric criteria to determine individual decisions.  Additionally we investigate the stability of a wealth distribution, both to local perturbations and to the introduction of new agents with no wealth.     
\keywords{ Equality of opportunity \and Intergenerational transmission \and Meritocracy }
\end{abstract}

\maketitle

\section{Introduction}

An important question about income distributions is the degree to which the income of a person depends on the income of their parents.  Typically, modeling intergenerational outcomes is done either in models comparing relative outcomes in terms of quartiles, deciles or percentiles, or in more detailed economic models that price labor and capital.  An example of the first would be using bi-stochastic matrices to model the probability that children born to parents in one quartile of the income distribution become a member of a different quartile \cite{kanbur2016dynastic}.  Examples of the second type are models that predict wealth distributions with heavy tails \cite{benhabib2016skewed}, micro-founded models that explicitly model individuals, a single firm and tax policies \cite{Donder2017}, and elegant models of intergenerational choices with complete dynamic programming solutions \cite{loury1981intergenerational}.  
While bi-stochastic matrices could seemingly apply to all societies, they are unable to capture the extent of inequality between quartiles. Meanwhile, the complexity of having to explicitly model economic forces typically limits the generalizability  of micro-founded models.  In contrast, we propose a moderately general model that can capture relative wealth differences and requires only qualitative modeling of economic forces and policy.

One benefit of a micro-founded model is it can capture how income distributions might affect investment decisions.  Indeed, the ability of a person to improve the outcomes of their children is one of the fundamental incentives for a person to work or invest during their own lifetime. Such an observation is occasionally used in popular discourse to argue against estate or inheritance taxes.
While the above arguments clearly rest upon a number of important empirical questions, they also involves a logical trade-off between effort/investment and intergenerational mobility. 
 In order to investigate the extent to which effort and/or investment can be induced through positive intergenerational outcomes, we construct a multi-generational model with rational actors conducive to the argument that intergenerational rewards compel effort and investment. 
 
Our model will include agents in multiple non-overlapping generations, where each agent has an endowment and dedicates a portion of that endowment towards the outcome of their offspring (as measured by the offspring's endowment) while consuming the remainder. Specifically, agents use a Cobb-Douglas utility function to balance the benefit of their immediate consumption with the endowment outcome of their offspring\footnote{Under this formulation, parents are guided by a factor of descendant utility, but not their utility directly, and thus the parents are not technically altruistic, as seems to be the case empirically \cite{AltonjiParentalAltruism}. }. We assume the endowment outcome of children is a monotonic function, $T$, of their parent's investment in them, but make no further assumptions on its shape. Notice, this allows for models where the offspring's outcomes have increasing marginal returns on parental investments, consistent with a world where the rate of return to capital grows with assets, as appears to be the case \cite{fagereng2016heterogeneity,loury1981intergenerational,piketty2014capital}.  
 In the limit of large populations, we prove that there is a fundamental trade-off between extreme levels of effort and mobility. Namely, we show that regardless of the shape of $T$, there is a limit on parental investment, such that if any agent invests more than this limit, then there must exist some lower level of endowment that an otherwise identical agent's dynasty would be forever trapped below. Interpreting the extra effort of investment above this limit as indicative of a `rat race' and the presence of intergenerational immobility between different strata as a `wealth trap' allows for this theorem to be restated as `the existence of a rat race implies the existence of a wealth trap'.

We propose this model as a more general framework through which to investigate the fundamental trade-offs between different intergenerational outcomes, the effects of social policy and how these combine to determine income distributions.  We also briefly mention how this model can make sense of some political stances and demonstrate that this model has some interesting behavior relevant for some notions of meritocracy.

\section{An Intergenerational Model}

Consider discrete generations $j$, where between generations we assume, for simplicity,  that each member of the society $i$ is replaced by a single descendant.   We denote the proportion of total resources, income and/or status of an individual $i$ during generation $j$ abstractly as endowment $w_{i,j}\in\mathbb{R}^+$.

We assume that each individual can impact their descendant's competitiveness through gifts, investments in education, bequests, or other efforts, which we represent as a fraction $x_{i,j}\in[0,1]$ of their own endowment.  Whatever is not invested in a descendant is consumed.  We describe the trade-off between consuming $(1-x_{i,j})w_{i,j}$ and investing $w_{i,j}x_{i,j}$ using a Cobb-Douglas utility function:
\begin{equation}
u_{i,j}(x_{i,j})  =  ((1-x_{i,j})w_{i,j})^{1-\alpha_i}  (w_{i,j+1}(w_{i,j}x_{i,j}))^{\alpha_i}, \label{CobbDouglas}
\end{equation} 
where $ w_{i,j+1}$ is the observed or anticipated endowment of $i$'s descendant in generation $j+1$, and $\alpha_i \in (0,1)$ determines $i$'s trade-off between these terms. In equations where all terms belong to the same generation, we will generally omit the subscript $j$.

How $w_{i,j+1}$ depends on the investment $w_{i,j}x_{i,j}$ determines the dynamics. During any generation, $i$'s endowment can be expressed as proportional with $i$'s competitiveness, $T_{i,j}$, and the total endowments available to generation $j$, $W_j$,  so that  $w_{i,j} = W_{j}T_{i,j}/ \sum_k T_{k,j}$.  If the society is non-discriminatory and the same opportunities for intergenerational investment are available to all, then $T_{i,j+1} = T_j(x_{i,j}w_{i,j})$. In principle, $w_{i,j}x_{i,j}$, the argument of $T_{i,j}$,  captures the extent of resources invested in a child's raising and education and/or those gifted to the child; while the value of  $T_{i,j}(w_{i,j}x_{i,j})$ determines the proportion of $W_{j+1}$ that $i$'s descendant eventually acquires. 

The form of $T_{i,j}(y)$ can capture many different dynamics: for instance if $T_{i,j}(y)=y$, then one assumes that an endowment invested in a child exactly linearly increases the child's endowment; if $T_{i,j}(y) = c$ for a constant $c$ then the endowment of a child is completely independent of the parent's endowment or investment; whereas a curve such as $T_{i,j}(y) = y(1+\frac{1}{4}\tanh (y))$ has inheritances with increasing rates of return.

\subsection{ Further Assumptions}

We assume that the primary components of relative competitiveness depend less on capital than on a complex interaction of  human instruction, mentoring and hierarchical status, that have costs which generally track economic growth.  For example, many of the inputs for human capital such as: tuition at universities and elite primary and secondary schools, and medical care are common examples of Baumol's cost disease in that their costs have tracked or exceeded, sometimes greatly, per capita economic growth.  Thus, we assume that understanding the relative differences in income can be achieved in a model where $T_j(y) = T(\frac{y}{W_j})$.  In such a situation, we simplify the system without losing any dynamics by assuming that $W_j=n$, the fixed population size. Even when this assumption on the determinants of relative competitiveness is not appropriate, the equilibrium analysis in this paper is still useful as a way of discussing properties of the current distribution and competitiveness function $T$. 

Investigating situations where $T_j$ and $W_j$ are functions of $\vec{x}$ and $\vec{w}$ is naturally interesting.  Indeed, under some assumptions of $W_j$ and $T_j$, there is a tension between maximizing the longterm utility and the stability of the income distribution, while under other assumptions such a tension is not present.  While intriguing, such questions are beyond the current scope of this paper.

Assuming that $W_j=n$ and that $T_j = T$ now allows us to focus on the shape of the competitiveness function $T$. 
Indeed, the shape of $T$ determines how different levels of intergenerational investment determine future competitiveness, and the shape is a concise summary of how economic, political and social systems inside a society determine and apportion outcomes.  Further, and as we shall see, the shape of $T$ determines the overall stability or instability of a given income distribution. We assume only that $T$ is non-decreasing and continuous.  Indeed, in many potential societies the relative payoffs may depend on complex interactions or discrete cutoffs. For example, the benefits of an expensive private school may be significant, but available only to those with sufficient means.

\subsection{Rational behavior}

As standard, we consider the scaled logarithm of the utility function in eqn. \ref{CobbDouglas} and omit constants, yielding:
\begin{equation}
U_{i,j}(x_{i,j}) = \ln(1-x_{i,j}) + A_i \ln(T(x_{i,j}w_{i,j}))- A_i\ln\left(\sum_k T(w_{k,j}x_{k,j})\right),
\end{equation} 
for $A_i = \frac{\alpha_i}{1-\alpha_i} \in (0,\infty) $. 
When $ T(w_{i,j}x_{i,j}) <<  \sum_k T(w_{k,j}x_{k,j})$, such as when the population is very large, then the term $A_i\ln\left(\sum_k T(w_{k,j}x_{k,j})\right)$ does not significantly impact rational individual decisions.  Thus, in the large population limit, rational maximization of the utility function implies that $x_{i,j}$ is a function of $w_{i,j}$, denoted as function $g_A(w_{i,j})$:
$$x_{i,j} = g_{A_i}(w_{i,j}) =  \mathrm{argmax}_{x\in[0,1]} \left[ \ln (1-x) + A_i \ln(T(xw_{i,j}))\right] .$$

For some functions $T$, it is possible to analytically calculate $x$.  For example,
if the transfer function is $T(y)=y$ and the population is homogenous with $A_i = A$ for all $i$, then in the large population limit $x_{i,j} = \frac{A}{A+1}$ and $w_{i,j+1} = w_{i,j}$, implying that the income distribution exactly reproduces itself.

  More generally, in a large uniform population, $T(y)=y^k$ leads to a unique solution: $x_{i,j} = \frac{kA}{kA+1}$ and $\frac{w_{i,j+1} }{w_{k,j+1} } = ( \frac{w_{i,j} }{w_{k,j} })^k $.  Thus, an entity in control of the shape of $T$ could, by increasing $k$, compel efforts arbitrarily close to $1$.  However, for $k>1$, the ratio of wealth in one generation is exacerbated in the next. Thus, for $k>1$, this transfer function inevitably results in a winner take all income distribution, where a single individual claims virtually all the endowment\footnote{As a single individual begins to control a large portion of the total endowment, the large population limit falls apart and the single oligarch's exertion decreases.}.

For more general $T(y)$, one can determine the equilibrium strategy graphically via Lagrange multipliers.  For some endowment $w_i$, consider the level sets of utility  $C = \ln{w_i}+ \ln{(1-x_i)} +A_i\ln(T(x_iw_i))$, which is achieved by $T^{A_i} (x_iw_i)=  \frac {e^C}{w_i(1-x_i)}  = \frac { C_2}{w_i(1-x_i)} $.  The optimal $x_i$ is thus the value of $x_i$ at the intersection between the curves $T^{A_i}(x_iw_i)$ and $\frac {C_2}{w_i(1-x_i)} $ for maximal $C_2$, as illustrated in Figure \ref{jumps}. Since $T$ is not assumed to be concave, then it is possible for the system behavior to have discontinuities when several levels of effort appear equally good, as in Figure \ref{jumps}.  If $T$ and $g_i$ are differentiable, then it must also be the case that $\frac{d T^{A_i}}{dx} = \frac {e^C}{(w_i-w_ix_i)^2}$, leading to the useful fact that at rational $x_i$: 
\begin{equation}
\frac{T'(w_ix_i)}{T(w_ix_i)} = \frac{1}{A_i(w_i-w_ix_i)},\label{t't}
\end{equation}  
which constrains the shape of the $T$.

\begin{figure}
\includegraphics[width=.9\linewidth]{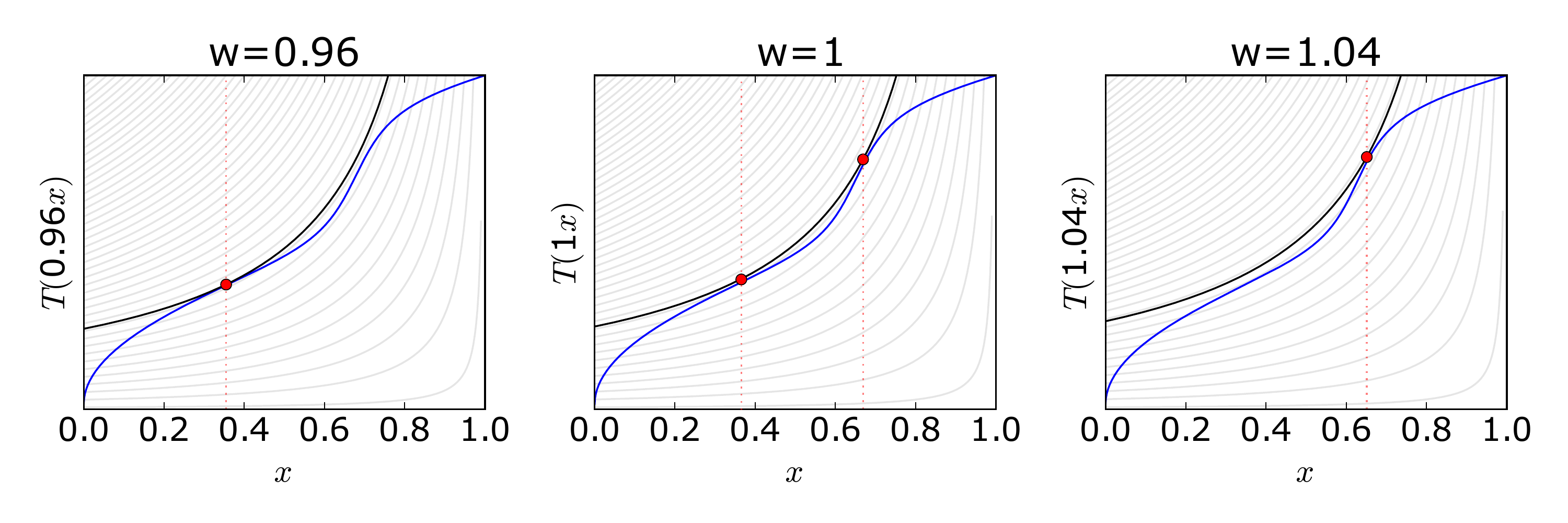} 
\caption{The optimal effort $x_i$ corresponds the maximal value of $C_2$ such that the curve $I(y)=\frac{C_2}{1-y}$ intersects $T^A(w_i y)$, where the value of $y$ at that intersection is the optimal effort. For different endowments, the relevant portion of the transfer curve changes, altering the optimal efforts. Complicated functions $T$ lead the dynamics to be discontinuous when multiple levels of effort give the same utility. 
} \label{jumps}
\end{figure}

\section{Wealth traps and meritocracies}

An important notion of stability for a wealth distribution is the stability of the distribution to sudden shocks to individual agents, or somewhat equivalently, the addition of new agents with very small endowments\footnote{We discuss stability to infinitesimal perturbations, i.e. linear stability, briefly in the appendix. }. In other words, if a member of the population, $i$, suffered an exogenous catastrophe that destroyed almost all of $i$'s endowment, $w_i \to \delta$, for some $0<\delta<<1$ (or if $i$ immigrated into the system with only a small initial endowment $\delta$) would the descendants of $i$ eventually recover, reproducing the original distribution, or would the descendants of $i$ be trapped in some lower social strata?  We formalize such a situation as a wealth trap:

\begin{definition}[Wealth Trap]
An equilibrium in the large population limit has a wealth trap at parameter value $A_0$ if, after adding new agent $i$ with $A_i=A_0$, $w_{i,0}=\epsilon$ $\epsilon>0$ there exists agent $k$, $A_k=A_i=A_0$ such that $\lim_{j\to \infty} w_{i,j} < w_{k,j}$
\end{definition}

For example, consider a stepwise transfer function with $T(y)=0.01$ for $y<0.5$ and $T(y)=1$ for $y\ge 0.5$ with egalitarian distribution $w_i = 1$ and corresponding $x_i = 0.5$ for all $i$. Notice that this system is locally stable to small perturbations in $w_i$, since the optimal strategy is always to invest exactly $w_ix_i = 0.5$.  However, this system has a wealth trap, since any person with wealth less than $0.5$ is unable to attain investment $ 0.5$ and instead achieves only $1/100$th the endowment of the remaining population for all time thereafter.

In contrast, the transfer function $T(y) = \sqrt{y}$ and the egalitarian endowment distribution $w_i =1 $ with uniform $A_i$ does not have a wealth trap. Indeed, as discussed earlier, endowments converge like $\frac{w_{i,j+1} }{w_{k,j+1} } = \sqrt{ \frac{w_{i,j} }{w_{k,j} }} $.

Notice that the definition of a wealth trap is specific to agents with the same parameter values of $\alpha_i$.  In fact, it is possible to have a situation where there is a wealth trap for agents with some parameters but not for their differently parametrized peers, as in Figure \ref{diverse_trap}.

A somewhat related notion is that of a meritocracy, which approximately sorts endowments $w_i$ according to investment appetites $\alpha_i$ :
\begin{definition}[Meritocracy]
An equilibrium in the large population limit is a meritocracy with respect to $\alpha$ if and only if for all $i$ and $k$, $\alpha_i \le \alpha_k$ implies $w_i \le w_k$ 
\end{definition}

Notice that the above definition of meritocracy applies to particular equilibria and not to a system more generally. In fact, for a given distribution of $\alpha_i$ and a choice of $T$, there can be subtle differences between the possible equilibria such that some equilibria are meritocracies while others are not.  A potentially unfortunate side-effect of defining meritocracies for equilibria, as opposed to for systems, is that it is possible for meritocratic equilibria to have wealth traps\footnote{Meritocracies can have wealth traps provided either no agent is stuck in a wealth trap, or all of the agents below some $\alpha_i$ are.}.  However, as we see in the next theorem, there remains a relationship between wealth traps and meritocracies. 

\begin{figure}
\begin{centering}
\includegraphics[width=.45\linewidth]{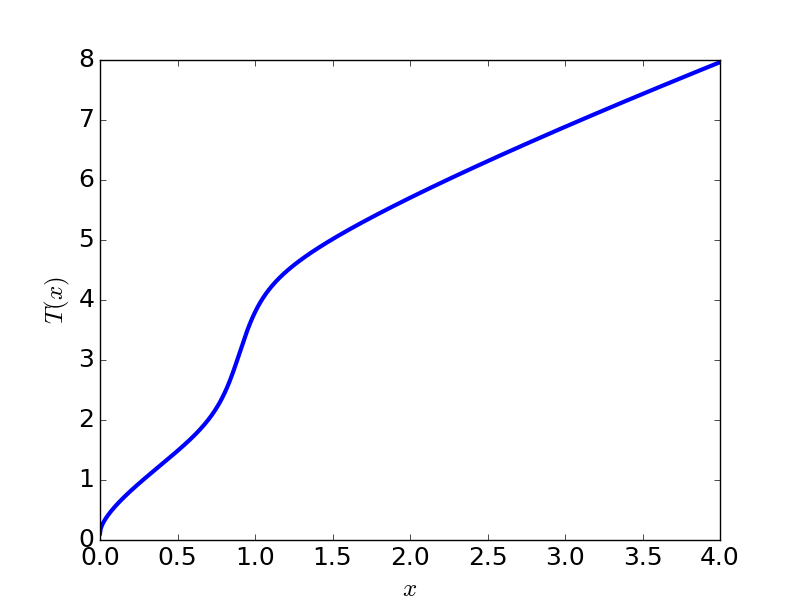} \includegraphics[width=.45\linewidth]{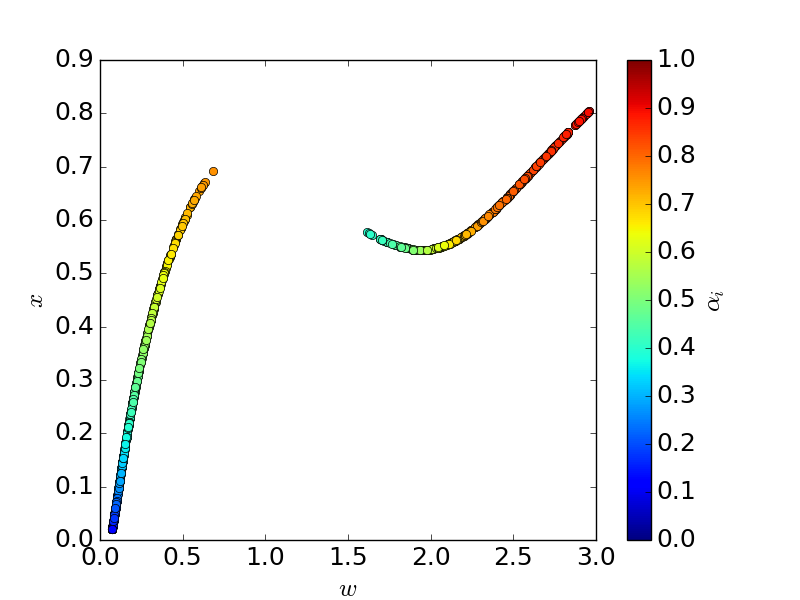}
\caption{ A system of agents with parameter $\alpha_i$ distributed uniformly at random with the above transfer function (left), has an equilibrium that sorts agents into two intervals of wealth (right).  Namely, agents with the highest and lowest values of $\alpha_i$ are naturally sorted to the lower and upper classes, while there exists a wealth trap for agents with intermediate values of $\alpha_i$, they can be found in either class. 
} \label{diverse_trap}\end{centering} 
\end{figure}

\begin{theorem}
In the large population limit and at some equilibrium, if there are no wealth traps, then the population is a meritocracy with respect to $\alpha$. 
\end{theorem}
\begin{proof}
If the system were not a meritocracy, then there exists $i$ and $k$ such that $\alpha_i\le\alpha_k$ but $w_i>w_k$.  However, since there are no wealth traps, then $i$ must be able to recover from a shock that reduces $w_i$ to $w_k$, but since $\alpha_k\ge \alpha_i$, then $x_k(w_k)\ge x_i(w_i)$ and thus $k$ must grow along with $i$'s recovery, implying that $k$ was not initially at equilibrium, a contradiction.  
\end{proof}

Clearly, though, any equilibria with a wealth trap neighbors equilibria where two identical individuals are on opposite sides of the trap, seemingly implying that the system is not a meritocracy (though in the large population limit, moving individuals changes the distribution on a set of measure zero). 

An important feature to remember about meritocracies is that it is not the case that if $w_i>w_k$ then $x_i>x_k$.  Namely, meritocracies with respect to $\alpha$ do not necessarily award endowment according to effort $x$. Indeed, equilibria in which $x_i\ge x_k$ implies $w_i\ge w_k$ would form a different type of meritocracy and may be worth future study.
 In any case, all equilibria clearly have the property that if $w_i>w_k$ then $w_ix_i>w_kx_k$.

\section{Effort as a function of endowment}

Policy arguments occasionally set as their aim a specific shape of the effort function $g_A(w)$.  However, not all shapes of $g_A(w)$ are possible, as we shall see in this section by momentarily reversing our perspective: rather than taking $T$ as granted and determining from it the effort function $g_A(w)$, we assume the effort function is known and attempt to infer $T$. This inference will allow us to to determine the conditions for the existence and uniqueness of $T$ for a given $g_A(w)$.  
We will show that there are only two limitations to the shape of a piecewise differentiable $g(w)$: the first limitation is that total investment, $wg(w)$, is non-decreasing in $w$; the second limitation is that $g'(w)$ is Lipshitz continuous at all $w$ where $g'$ exists.  These two conditions allow for the slope of $T$ to be integrated, thus creating a suitable $T$ that is unique to within rescaling. Integrating $T$ is also useful in situations where $g$ can be empirically estimated.

\begin{lemma}
For any transfer function $T$ and some value of $A$, investment $wg_A(w)$ is non decreasing as a function of $w$. \label{investmentNonDecreasing}
\end{lemma}
\begin{proof}
Suppose that at two endowments, $w_1$ and $w_2$, have $w_1g(w_1)=y_1$ and $w_2g(w_2)=y_2$.  For $y_1$ to be optimal at $w_1$, it must be that if for some $c_1$, $T^A(y_1)=\frac{c_1}{w_1-y_1}$ and  $T^A(y_2)\le \frac{c_1}{w_1-y_2}$. 
Similarly, there is $c_2$ such that $T^A(y_2)=\frac{c_2}{w_1-y_1}$ and  
$T^A(y_1)\le \frac{c_1}{w_2-y_1}$.  Manipulating these two inequalities implies that $\frac{w_2-y_1}{w_2-y_2} \le \frac{w_1-y_1}{w_1-y_2}$ and thus $(y_2-y_1)(w_2-w_1)\ge 0$, implying that $y_2>y_1$ only if $w_2\ge w_1$.
\end{proof}

That $wg$ is non-decreasing implies that $g(w_0) \le \lim_{w\to w_0^+}g(w)$ and, subsequently, that about any $w_0$ there is an interval on which $g(w)<1-\epsilon$ for $\epsilon>0$.

Next, consider how to integrate $T(y)$ at points where $g(w)$ is differentiable.  Notice that since at any $w$ there exists $c$ such that, $T^A(wg(w))) = \frac{c}{w-wg(w)}$ and tangency at this point gives that $AT^{A-1}\frac{dT}{dy}(wg(w)) = \frac{c}{(w-wg(w))^2}$, then,
\begin{eqnarray*}
\frac{d}{dw} T(wg(w)) &=& (g(w)+wg'(w))\frac{dT}{dy}(wg(w)) \\ 
&=& \frac{c(g+wg')}{A(w-wg)^2T^{A-1}} \\
&=& \frac{T(wg)}{A}\frac{(g+wg')}{(w-wg)} \\
\end{eqnarray*}
Since $wg(w)$ is non decreasing and $g(w)<1-\epsilon$, $T$ can be integrated forward or backward and is non-decreasing.  Further, if $g'$ is Lipshitz continous everywhere, then $\frac{d}{dw} T(wg(w))$ is as well, implying that $T$ is unique to within scaling.  

Discontinuities in $g(w)$ represent points when there are two levels of investment that return equal utilities. Thus, if there is at some $w_0$, $\lim_{w\to w_0^-}g(w) = x^-$ and  $\lim_{w\to w_0^+}g(w) = x^+$ where $x^- \ne x^+$, then for some c $\lim_{w\to w_0^-}T^A(wx^-) = \frac{c}{w_0-w_0x^-}$ and $\lim_{w\to w_0^+}T^A(wx^+) = \frac{c}{w_0-w_0x^+}$.  Thus, $(1-x^-)\lim_{w\to w_0^-}T^A(wx^-) =(1-x^+) \lim_{w\to w_0^+}T^A(wx^+)$, allowing for the integration to be continued.  In such a scenario, notice that $T(y)$ can assume any value less than $\frac{c}{w_0-y}$ between $w_0x^-$ and $w_0x^+$, and thus preclude uniqueness in $T$, though in a completely non-consequential manner.

\section{Rat races imply wealth traps}

While with different choices of $T$ it is possible to compel any level of effort less than $1$, greater levels of effort eventually imply the existence of a wealth trap. Indeed, as we will show in this section, this tradeoff between effort and intergenerational mobility is described by the following theorem:

\begin{theorem}\label{ratRace}
For a fixed transfer function $T$, if at equilibrium and in the large population limit there exists an agent $i$ such that $x_i>\alpha_i$, then there exists a wealth trap for those with parameter $\alpha=\alpha_i$.  
\end{theorem}

If we colloquially refer to exertion above $\alpha_i$ as a `rat-race,' then theorem \ref{ratRace} implies that a rat-race implies a wealth trap (or a `rat-trap').  Indeed, if the common notion of a rat-race is of a person toiling with little time for leisure, aware that falling behind in the rat-race would leave them and their descendants left forever behind, then this theorem seems entirely consistent with that depiction. We provide the necessary but rather tedious details to extend this theorem to potentially discontinuous $g_A(w)$ in the appendix.  

In order to simplify the following statements, we change variables from endowment $w$ to  $z = \ln{w}$, the order of magnitude of an endowment.  Similarly, we let $\bar{g}(z) = g(e^z)$, and at equilibrium we denote $\gamma = \frac{W}{ \sum_k T(w_{k,j}x_{k,j})}$ so that $w_{i,j}= e^{z_{ij}} = \gamma T_{i,j}$ in the large population limit.

\begin{proof}[of theorem \ref{ratRace} for differentiable $g$]

Assume to the contrary that there are no wealth traps and endowment level $z^*$ achieves $g_A(z^*) =\alpha_i+\epsilon =  \frac{A}{A+1}+\epsilon$ for some $\epsilon>0$.  We will arrive at a contradiction by attempting to integrate the inter-generational rate of return, $r_{i,j} = \frac{w_{i,j+1}}{x_{i,j}w_{i,j}}$, backwards from $z^*$ without producing a wealth trap. 
  As we integrate backwards, we will establish two facts: first, $\bar{g}(z)$ remains at or above $ \frac{A}{A+1}+\epsilon$; second, either $r(z)$ decreases or $\bar{g}(z)$ increases as $z$ decreases.  It will thus follow that while integrating backwards, one of two things must inevitably happen: either $g_A(z)$ exceeds its bound of $1$, or $r$ becomes less than $1$, implying a wealth trap.

The condition that the system is at equilibrium at $z^*$ implies $r(z^*) = \frac{w_i}{ w_i\bar{g}(z^*)}= \frac{1}{ \bar{g}(z^*)}\le \frac{A+1}{A +\epsilon(A+1)}$.  Further, whenever $r(z)\le \frac{A+1}{A +\epsilon(A+1)}$, then $\bar{g}(z)\ge \frac{w_{j+1}}{w_j}(\frac{A}{A+1} +\epsilon)>\frac{A}{A+1} +\epsilon$, provided there are no wealth traps.  Thus, 

\begin{eqnarray*}
\frac{\partial r}{\partial z} &=& \frac{d }{dz} \frac{T(e^z\bar{g}(z))\gamma}{e^z\bar{g}} \\ 
&=& \gamma (1 + \frac{\bar{g}'}{\bar{g}})\left(\frac{dT}{dy}-\frac{T}{e^z\bar{g}}\right)  \\ 
&=& \gamma T (1+\frac{\bar{g}' }{\bar{g}})\left( \frac{c}{AT^A (e^z-e^z\bar{g})^2} -\frac{1}{e^z \bar{g}} \right)   \\ 
&=& \gamma T (1+\frac{\bar{g}' }{\bar{g}})\left( \frac{1}{A e^z(1-\bar{g})} -\frac{1}{e^z \bar{g}} \right)   \\ 
&=& r\frac{1}{1-\bar{g}} (1+\frac{\bar{g}' }{\bar{g}})\left( \frac{A+1}{A}\bar{g}-1  \right)   \\
&\ge & \epsilon(1+\frac{\bar{g}' }{\bar{g}})
\end{eqnarray*}
Since $\frac{d}{dz}e^z \bar{g}(z) = e^z \bar{g} +e^z \bar{g}' \ge 0$, then lemma \ref{investmentNonDecreasing} also implies that $\frac{\bar{g}'}{\bar{g}} \ge -1$, implying that the factor $(1+\frac{\bar{g}' }{\bar{g}})\ge0$.  Thus, at any $z<z^*$, either $\bar{g}' =-\bar{g}$ or $\frac{\partial r}{\partial z}>0$, and integrating from $z^*$ to $-\infty$ leads to either $\bar{g}>1$ or $r<1$ at some $z$, either one a contradiction.  

\end{proof}


\section{Discussion}

We have shown that in any large fixed society where intergenerational outcomes determine effort according to a Cobb-Douglas utility function, extraordinary levels of effort imply a wealth trap.  For any single transfer function $T$, this relationship may seem obvious, but the result establishes the same bound for all possible shapes of a fixed $T$.  Thus, no amount of clever redesigning of $T$ can maximize efforts above $\alpha_i$ while avoiding the disegalitarian properties of a wealth trap.  While the results of this paper are specific to the model, the more abstracted point likely generalizes: that encouraging work with increased rewards eventually becomes synonymous with encouraging work with an implicit threat of someone being `left behind.'   

Of course, as is with models, many of the assumptions are somewhat naive. 
An assumption worth considering is that $T$, to the extent that such a function even exists, is surely not fixed in time.  Indeed, policy that changes $T$, such as subsidies/funding for education, inheritance tax rates, social insurance, etc., are some of the most politically contentious issues and regularly change within a single generation. Ironically though, that these issues are so frequently debated only increases the need for theoretical frameworks with which to gauge claims of their effects, and we believe equilibrium analysis remains useful for this task.  

Specifically, this model allows us to readily observe the likely arguments of different people throughout the income distribution. Clearly, every person $i$ has self-interested reasons to argue for increasing $T(y_i)$ and decreasing $T$ at all other values. However, while the self-interest is apparent, any political arguments for changing $T$ must necessarily be made in a way that try to appeal to rest of the population (or at least a majority). Our model assumes $T$ is known, but that the total endowment function, $W(\vec{x},\vec{w})$ is not, therefore the natural way to structure arguments is to argue about the shape of $W$ around the current values of $\vec{x}$, $\vec{w}$.  
 We briefly discus how these self-interested arguments may be structured at different endowment levels below, and we note that they appear at exactly as one expects and they are the same as or similar to existing arguments. 

Individuals near the top of the endowment distribution with endowments $w_+ = \max_k w_k$ have naturally self-interested  arguments to
claim that the marginal impact of effort, $\frac{\delta W}{\delta x}$, is large.  Arguments that $\frac{\delta W}{\delta x}$ are large, such as claims that classes of people do not work enough, suggest that the curvature of $T$ should be increased so as to induce larger choices of $x$, and integrating this change over all of $T$ necessitates a large increase to $T(w_+g(w_+))$. Notice that this line of reasoning would be particularly critical of welfare, or public services, which leads to $T(0)>0$, since any non-zero intercept leads $x(0)=0$. 
A complementary argument to the focus on $x_i$ is that $W$ is sensitive to increasing $x_iw_i$, and thus increasing $x_i$ for the well endowed is particularly important.  These arguments are precisely the set of arguments that claim that a tax on inheritance, or taxes in general, decrease savings beyond what is optimal.  

The above arguments are centered on claims of large partial derivatives of $W$ and are thus ultimately arguments framed around Pareto improvements and greater rates of growth. 
Regardless, the self-interested arguments of the well-endowed inevitably and unsurprisingly push $T$ towards greater inequality. 
 Notice that while there is no natural limit to the extreme values of $x$ for the above argument, 
 theorem \ref{ratRace} gives that this line of argumentation has a fundamental limit, beyond which inducing greater effort leads to discreet and immobile classes. 

Meanwhile, although much of this model was explicitly designed to explain the possible arguments of the well-endowed, natural arguments are available to the rest of the endowment distribution.  

The first such argument is that $\frac{\delta W}{\delta x}$ is small or negative.  In particular, if the primary problem facing the society were a lack of aggregate demand, then not only would $W$ increase, the total social utility, which depends on $\sum_k (1-\alpha_k)\log (w_k)$, would be increased by decreasing the curvature of $T$, reducing $x_i$.  Similarly, utilitarian concerns more generally support reducing $T(y)$ for large $y$ and increasing it for small $y$, especially very small $y$.  

Notice that since the above arguments for flattening $T$ are likely distributional, only for very negative values of $\frac{\delta W}{\delta x}$ would they be a Pareto improvement.

As has been observed before, those that claim that the principle goal of analysis should focus solely on Pareto gains will be biased towards the self-interested arguments of the well-endowed and biased against the arguments of the lesser-endowed.  

The above analysis assumes that the natural division of interests are between those above some endowment threshold  (e.g. the median income) and those below it.  It is worth considering whether there could be a natural politics of the middle class, where the middle had natural economic cause against a coalition of the well- and lesser-endowed.  The answer is, of course, that such a division is possible, but requires either a more complicated argument about $W$ or perhaps more naturally represents balancing the above arguments.

Ultimately, this model of intergenerational mobility is able to support class-based arguments, which is probably not a difficult feat, but remains interesting nonetheless.

\section{Conclusions}

In this paper we have proved the first of hopefully many `rat-race' theorems, which establish a fundamental connection between extraordinary levels of effort and distributional issues.  Clearly, important questions can be addressed by extending the model developed in this paper as well as proving similar theorems in completely new models.  

In terms of extensions of the model developed in this paper, adding noise to the model, such that $T(y)$ provided the mean of some distribution of outcomes rather than the deterministic outcome, would be a natural direction for future work.  Depending on the variance in a stochastic setting, deterministic wealth traps could be anything from barriers that a dynasty might eventually overcome, barriers that require exponential time to overcome, to completely inconsequential features in a system governed by the stochastics. Determining at what levels of noise the appropriate summary of the system ceases to be the deterministic one presented here would be an important accomplishment. It would also be interesting to determine, when fitting a distribution, how much leeway exists between stochastic parameters and the deterministic parameters governing $T$.  

Another potential direction for future research is to begin to include other heterogeneities into the population.  Currently the model can accommodate a notion of meritocracy based upon an intergenerational discounting rate, $\alpha$, but including individualized competencies would be a natural way to expand the model. Ultimately, in this an other settings, we hope that this paper furthers and broadens the examinations into the connections between income mobility and social policy.

\section{Appendix}

\subsection{Linear stability}
Here we derive several useful formulas. First, we use the fact that $\frac{T'}{T} = \frac{1}{A_i(w-wg_i(w))}$ to find how investment decisions change with $w$.
\begin{eqnarray*}
\frac{\partial g(w)}{\partial w } &=& -\frac{U_{xw}}{U_{xx}}\\
&=& -\frac{Axw(TT''-T'^2)+AT'T} {-T^2(1-g)^{-2}+Aw^2T''T-Aw^2T'^2} \\
&=& -\frac{g}{w} + \frac{1}{w(1+\frac{1}{A}-Aw^2(1-g)^2\frac{T''}{T})}
\end{eqnarray*}
This implies that the change in total investments follows:
\begin{eqnarray*}
\frac{\partial (wg(w))}{\partial w } &=& wg'(w) +g(w)\\
&=& \frac{1}{(1+\frac{1}{A}-Aw^2(1-g)^2\frac{T''}{T})}.
\end{eqnarray*}
 Further, since investment cannot fall with endowment, this quantity must be positive, and thus $\frac{T''}{T}< \frac{A-1}{A^2(w-wg(w))^2}$, which is the same condition that $T^A$ has a smaller second derivative than a curve of constant Utility.  

This equation can then be used in the equations for linear stability.  If $g_i$ and $T$ are differentiable, the standard linear stability of an equilibrium is given by the eigenvalues of the Jacobian with entries:
\begin{eqnarray*}
\frac{\partial w_{i,j+1}}{\partial w_{i,j}} &=& \gamma T'  (w_{i,j}g_i(w_{i,j})) \left( w_{i,j}\frac{\partial g_{i}}{\partial w} (w_{i,j}) + g_i(w_{i,j}) \right)   \left( 1 - \frac{w_{i,j}}{W_j}  \right)  \\
&=& \frac{  \gamma A ( 1 - \frac{w_{i,j}}{W_j}  )TT' } {T+TA-A^2w_{i,j}^2(1-g_i(w_{i,j}))^2T''} \\
\frac{\partial w_{k,j+1}}{\partial w_{i,j}} &=& \gamma T'  (w_{i,j}g_i(w_{i,j})) \left( w_{i,j}\frac{\partial g_{i}}{\partial w} (w_{i,j}) + g_i(w_{i,j}) \right)   \left( -\frac{w_{k,j}}{W_j}  \right) \\
&=& \frac{  -\gamma A\frac{w_{i,j}}{W_j}  )TT' } {T+TA-A^2w_{i,j}^2(1-g_i(w_{i,j}))^2T''} \\
\end{eqnarray*}
Since total endowment changes are zero-sum, column sums add to zero. Naturally, this implies a zero eigenvalue corresponding to the neutrally stable perturbation, which enriches everyone so that the same distribution recreates itself.  For the special case of fully egalitarian distributions, where $w_i=w_k$ for all $i$ and $k$, the Gershgorin circle theorem implies that the equilibrium is stable if $\gamma T'  (w_{i,j}g_i(w_{i,j})) \left( w_{i,j}\frac{\partial g_{i}}{\partial w} (w_{i,j}) + g_i(w_{i,j}) \right)<1  $.

\subsection{Proving theorem \ref{ratRace} for discontinuous $\bar{g}$}

Note that first, in considering more general $\bar{g}$, we need not worry about $z$ where $\bar{g}(z)$ is continuous but not differentiable, as wherever $\bar{g}$ is continuous so is $r$ and the monotonicity of $T$ implies that such points occupy a set of Lebesgue measure $0$.  Thus, points where $\bar{g}$ is continuous but not differentiable cannot affect an integration of $\bar{g}$ or $r$.

Finally, we first investigate how $r$ changes with $z$ at discontinuities of $\bar{g}(z)$ when $r<\frac{A+1}{A}$.   In order to extend the proof in the main body, we need to show that for decreasing $z$, discontinuities in $g$ do not prevent the previous contradiction; either $r$ decreases below $1$ or $x$ increases past $1$.  As we will show,  discontinuities in $\bar{g}$ cannot increase $r$, and while they can decrease $\bar{g}$, they cannot do so without also changing $T^A$.

\begin{proof}
Consider some point $z_0$ such that $\lim_{z\to z^{_0-}} \bar{g}(z)= x_0$ and $\lim_{z\to z^{_0+}}\bar{g}(z) = x_1$, where $1<r(z_0)<\alpha^{-1}$ and $x_1>\alpha$. 

First, we show that $x_0>\frac{A}{A+1}$.  Since $x_0$ and $x_1$ are both optimal, then 
there exists some $c$ such that $T^A(x_1w) = \frac{c}{w(1-x_1)}$ and $T^A(x_0w) = \frac{c}{w(1-x_0)}$.  Manipulating these equations and the assumption that $r=\frac{\gamma T(x_0w)}{x_0w}<\frac{A+1}{A}$ and $w \le \gamma T(x_1w)$ yields: 
\begin{eqnarray*}
x_0 &=& 1 - \frac{c}{wT^A(x_0w)} \\
&=& 1 - \frac{w(1-x_1)T^A(x_1w)}{wT^A(x_0w)} \\
&>&  1-\frac{w^A}{\gamma ^A T^A(x_0w)} (x_1^A(1-x_1))\left( \frac{A+1}{A} \right)^A   \\
&\ge& 1- x_1^A(1-x_1)\left( \frac{A+1}{A} \right)^A\\
\end{eqnarray*}
The minimum of this quantity is $x_0 >  \frac{A}{A+1}$, which is attained when $x_1 = \frac{A}{A+1}$. 

To see that $r_0<r_1$, consider:
\begin{eqnarray*}
\frac{r_0}{r_1} &=&   \frac{\gamma T(x_0w)}{\gamma T(x_1w)}\frac{x_1w}{x_0w}  \\
&=& \frac{x_1(1-x_1)^\frac{1}{A}}{x_0(1-x_0)^\frac{1}{A}}.
\end{eqnarray*}
Since $x^A(1-x)$ is decreasing for $x>\frac{A}{A+1}$, then $\frac{A}{A+1}<x_0<x_1$ implies $r_0<r_1$.   Thus, if $\bar{g}$ has only a finite number discontinuities or if $g$ has an infinite number of discontinuities but $\bar{g}' >-1 $ and thus $\frac{dr}{dz}>0$ on a set of infinite measure, these discontinuities do not prevent integration eventually leading to $g(w)>1$ or $r<1$.

In contrast, consider the scenario where $\bar{g}$ has an infinite number of discontinuities, $\frac{dr}{dz}>0$ on a set of finite measure and $\bar{g}'=-1$  almost everywhere else.  In order for $\bar{g}=-1$ on an infinite set but $\bar{g}>0$ always, then there must be an infinite number of discontinuities of $\bar{g}(w_k)$ such that the $\sum_k \bar{g}(w_k^+)-\bar{g}(w_k^-) = \infty$.  However, notice that at any discontinuity of $\bar{g}$:
\begin{eqnarray*}
T^A(wx_1)-T^A(wx_0) &=&\frac{c}{(1-x_1)}-\frac{c}{(1-x_0)} \\
&=& \frac{c}{w}\frac{x_1-x_0}{(1-x_1)(1-x_0)} \\
&\ge& \frac{c}{w}(A+1)^2(x_1-x_0)\\
\end{eqnarray*}
Thus, $T^A$ decreases at least proportionally to $\bar{g}$ as $z$ is integrated backwards, implying that if   $\sum_k \bar{g}(w_k^+)-\bar{g}(w_k^-) = \infty$, then $\lim_{y\to0}T^A(y) = -\infty$ contradicting the fact that $T>0$.  

\end{proof}

\section{Acknowledgments}

This work was improved by comments from Sian Mooney, Erika Camacho and Jeannie Wilson.

\bibliographystyle{plain}
\bibliography{incomeDistRef}

\end{document}